\documentclass[runningheads]{llncs}
\usepackage{microtype}
\usepackage{complexity}
\usepackage{mathtools}
\usepackage{amssymb}
\usepackage{microtype}
\usepackage[inline]{enumitem}
\usepackage{tikz}
\usepackage{pgfplots}
\usepackage[caption=false]{subfig}
\usepackage{algorithm}
\usepackage{algpseudocode}
\usepackage{caption}
\pgfplotsset{compat=1.16}
\usepgfplotslibrary{groupplots}

\newtheorem{observation}{Observation}
\newenvironment{proof-sketch}{\paragraph{Proof sketch:}}{\hfill$\square$}

\newcommand{\sample}{\ensuremath{\mathcal{S}}}

\newcommand{\size}[1]{\abs{#1}}

\newcommand{\prop}{\ensuremath{\mathcal{P}}}
\newcommand{\operators}{\ensuremath{\Lambda}}

\newcommand{\ltrue}{\mathit{true}}
\DeclareMathOperator{\lfalse}{\mathit{false}}

\DeclareMathOperator{\lnext}{\mathbf{X}}
\DeclareMathOperator{\luntil}{\mathbf{U}}
\DeclareMathOperator{\leventually}{\mathbf{F}}
\DeclareMathOperator{\lglobally}{\mathbf{G}}
\DeclarePairedDelimiter{\floor}{\lfloor}{\rfloor}

\DeclareMathOperator{\unaryvar}{\circ}
\newcommand{\binaryvar}{\mathbin{\circ}}

\newcommand{\ltlset}{\ensuremath{\mathcal{F}_\mathrm{LTL}}}

\newcommand{\unaryOp}{\ensuremath{\Lambda_U}}
\newcommand{\binaryOp}{\ensuremath{\Lambda_B}}
\newcommand{\sketch}{\ensuremath{\varphi^{?}}}
\newcommand{\placeholderset}{\ensuremath{\Pi}}
\newcommand{\placeholder}{\ensuremath{?}}
\newcommand{\suf}[1]{\mathit{suffix}({#1})}

\newcommand{\sketchprob}{\textit{LTL sketching}}
\newcommand{\learnprob}{\textit{LTL learning}}
\newcommand{\existprob}{\textit{LTL sketch existence}}

\newcommand{\sattable}[2]{\ensuremath{T}^{#1}_{#2}}
\newcommand{\existFormula}{\Phi^{\sketch,\sample}}
\newcommand{\incrementalFormula}[1]{\Psi^{\sketch,\sample}_{#1}}

\newcommand{\tool}{\texttt{LTL-Sketcher}}
\newcommand{\zeroset}{Type-0}
\newcommand{\onetwoset}{Type-1-2}
\newcommand{\algonesat}{Algo1-SAT}
\newcommand{\algonedt}{Algo1-DT}
\newcommand{\algtwo}{Algo2}

\newcommand{\abs}[1]{\ensuremath{|#1|}}

\newcommand{\nat}{\ensuremath{\mathbb{N}}}
\DeclareMathOperator{\lcm}{\mathit{lcm}}

\begin{document}
	\title{Specification sketching for Linear Temporal Logic}

	\author{Simon Lutz\inst{1} \and Daniel Neider\inst{2} \and Rajarshi Roy
		\inst{2}
	}
	
	\authorrunning{Lutz et al.}
	
	\institute{
		Technical University of Kaiserslautern, Kaiserslautern, Germany \and
		Max Planck Institute for Software Systems, Kaiserslautern, Germany
	}

	\maketitle            
	
	\begin{abstract}
		Virtually all verification and synthesis techniques assume that the formal specifications are readily available, functionally correct, and fully match the engineer's understanding of the given system.
		However, this assumption is often unrealistic in practice: formalizing system requirements is notoriously difficult, error-prone, and requires substantial training.
		To alleviate this severe hurdle, we propose a fundamentally novel approach to writing formal specifications, named \emph{specification sketching for Linear Temporal Logic (LTL)}.
		The key idea is that an engineer can provide a partial LTL formula, called an LTL \emph{sketch}, where parts that are hard to formalize can be left out.
		Given a set of examples describing system behaviors that the specification should or should not allow, the task of a so-called \emph{sketching algorithm} is then to complete a given sketch such that the resulting LTL formula is consistent with the examples.
		We show that deciding whether a sketch can be completed falls into the complexity class NP and present two SAT-based sketching algorithms.
		We also demonstrate that sketching is a practical approach to writing formal specifications using a prototype implementation.
		
	\end{abstract}
	

\section{Introduction}\label{sec:introduction}

Due to its unique ability to prove the absence of errors mathematically, formal verification is a time-tested method of ensuring the safe and reliable operation of safety-critical systems.
Success stories of formal methods include application domains such as communication system~\cite{FeckoUASDMM00,Lowe96}, railway transportation~\cite{BacheriniFTZ06,BadeauA05},
aerospace~\cite{GarioCMTR16,CoferM14}, and operating systems~\cite{VerhulstJ07,KleinAEHCDEEKNSTW10} to name but a few.

However, there is an essential and often overlooked catch with formal verification: virtually all techniques assume that the specification required for the design or verification of a system is available in a suitable format, is functionally correct, and expresses precisely the properties the engineers had in mind.
These assumptions are often unrealistic in practice.
Formalizing system requirements is notoriously difficult and
error-prone~\cite{Bowen20,PakonenPBV16,SchlorJW98}.
Even worse, the training effort required to reach proficiency with specification languages can be disproportionate to the expected benefits~\cite{Courtois}, and the use of formalisms such as temporal logics require a level of sophistication that many users might never develop~\cite{Holzmann02}.

We address this severe, practical issue by introducing a fundamentally novel approach to writing formal specifications, named \emph{specification sketching}.
Inspired by recent advances in automated program synthesis~\cite{SolarLezama13,SolarLezamaRBE05}, our new paradigm allows engineers to express their high-level insights about a system in terms of a partial specification, named \emph{specification sketch}, where parts that are difficult or error-prone to formalize can be left out.
Moreover, the engineer is asked to provide example executions of the system that the specification should allow or forbid.
Based on this additional data, a so-called \emph{sketching algorithm} then fills in the missing low-level details to obtain a complete specification.

While the concept of specification sketching is conceivable for a wide range of specification languages, we here focus on Linear Temporal Logic (LTL)~\cite{Pnueli77}.
The rationale for this choice is twofold.
First, LTL is popular in academia and widely used in industry~\cite{Fix08,GarioCMTR16,Holzmann97}, making it the de~facto standard for expressing (temporal) properties in verification and synthesis.
Second, LTL is well-understood and enjoys good algorithmic properties~\cite{ClarkeE81,Pnueli77}.
Furthermore, its intuitive and variable-free syntax have recently prompted efforts to adopt LTL (over finite words) also in artificial intelligence (e.g., to specify reward functions in reinforcement learning~\cite{DBLP:conf/ijcai/CamachoIKVM19}).

In the context of specification sketching for LTL, a sketch can leave logical operators or even entire subformulas unspecified, while examples are ultimately-periodic words.
To illustrate this setting, let us consider the request-response property $P$ expressing that every request $p$ has to be answered eventually by a response $q$.
This property can be formalized by the LTL formula $\varphi \coloneqq \lglobally (p \rightarrow \lnext \leventually q)$, which uses the standard temporal modalities $\leventually$ (``finally''), $\lglobally$ (``globally''), and $\lnext$ (``next'').
However, let us assume for the sake of this example that an engineer is unsure how to formalize $P$.
In this situation, our sketching paradigm allows the engineer to express high-level insight in the form of a sketch, say $\lglobally (p \rightarrow {\placeholder})$, where the question mark indicates which part of the specification is missing.
Additionally, the engineer provides two examples:
\begin{enumerate*}[label={(\roman*)}]
	\item a positive (infinite) trajectory $(\{ p\} \{ q \})^\omega$, expressing that $q$ is the response that should be used to answer a request, and
	\item a negative (infinite) trajectory $\{ q \}^\omega$, intended to disallow the system to send responses without requests.
\end{enumerate*}
Our sketching algorithm then computes a substitution for the question mark such that the completed LTL formula is consistent with the examples (e.g., $\lnext \leventually q$).
We set up all necessary definitions in Section~\ref{sec:preliminaries} and formally define LTL sketching in Section~\ref{sec:problem}.

It turns out that it is not always possible to find a substitution that is consistent with the given examples (see Example~\ref{eg:formula-non-existence} on Page~\pageref{eg:formula-non-existence}).
However, we show in Section~\ref{sec:theoritical-results} that the problem of deciding whether such a substitution exists is in the complexity class \NP.
Moreover, we develop an effective decision procedure that reduces the original question to a satisfiability problem in propositional logic.
This reduction permits us to apply highly-optimized, off-the-shelf SAT solvers to check whether a consistent substitution exists.

In Section~\ref{sec:algorithms}, we develop two sketching algorithms for LTL.
The first algorithm uses the decision procedure of Section~\ref{sec:theoritical-results} as a sub-routine and transforms the sketching problem into a series of LTL learning problems (i.e., in problems of learning an LTL formula---without syntactic constraints---from positive and negative examples).
This transformation allows us to apply a diverse array of learning algorithms for LTL, which have been proposed during the last five years~\cite{NeiderG18,CamachoM19,Riener19}.
Our second sketching algorithm, on the other hand, extends the LTL learning algorithm by Neider and Gavran~\cite{NeiderG18} and uses a SAT-based technique as an effective means to search for solutions of increasing size.

Finally, Section~\ref{sec:experiments} presents an empirical evaluation with a prototype implementation, named \tool.
We observed that our algorithms are effective in completing a variety of sketches having different missing information.
Further, comparing our algorithms, we concluded that our second algorithm outperforms our first algorithm in terms of running time and size of the solution.
We conclude this paper in Section~\ref{sec:conclusion} with a discussion of promising directions for future work.

\subsubsection*{Related Work}

The general idea of allowing partial specifications is not entirely new, but it has not yet been investigated as general as in this work.
The most similar setting is one in which templates are used to mine temporal specifications from system executions.
In this context, a template is a partial formula, similar to a sketch.
However, unlike a sketch, a template is typically completed with single atomic propositions or simple formulas, usually without temporal modalities (e.g., restricted Boolean combinations of atomic propositions).
An example of this approach is Texada~\cite{LemieuxB15,LemieuxPB15}, a specification miner for LTL\textsubscript{f} formulas (i.e., LTL over finite horizon).
Texada accepts an arbitrary template (property type in their terminology) and a set of system executions as input and completes the template with atomic propositions such that the resulting LTL formula satisfies all of the system executions.
Since LTL\textsubscript{f} (and LTL for that matter) is defined over a finite, user-provided set of atomic propositions, there are only finitely many ways to fill in the missing parts of a template, and Texada amounts to a search over this finite search space.
By contrast, our approach can complete a sketch with complex LTL formulas and, hence, has to search in an infinite space.

Various other techniques operate in settings where the templates are even more restricted.
For example, Li et al.~\cite{LiDS11} mine LTL specification based on templates from the GR(1)-fragment of LTL (e.g., $\lglobally\leventually\placeholder$, $\lglobally(\placeholder_1 \rightarrow \lnext\placeholder_2)$, etc.), while Shah et al.~\cite{ShahKSL18} mine LTL formulas that are conjunctions of a set of common temporal properties as identified by Dwyer et al.~\cite{DwyerAC98}.
In addition, the work by Kim et al.~\cite{KimMSAS19} uses a set of interpretable LTL templates, widely used in the development of software systems, to obtain LTL formulas robust to noise in the input data.
In the context of CTL, on the other hand, Wasylkowski and Zeller~\cite{WasylkowskiZ11} mine specifications using templates of
the form $\mathbf{A} \leventually \placeholder$, $\mathbf{A} \lglobally(\placeholder_1 \rightarrow \leventually\placeholder_2)$.
However, all of the works above complete the templates only with atomic propositions (and their negations in some cases).

Another setting is the one in which general (and complex) temporal specifications are learned from system executions without any constraint on the structure of the specification.
The most notable work in this setting is that by Neider and Gavran~\cite{NeiderG18}, who learn LTL formulas from system executions using a SAT solver.
Similar to their Neider and Gavran is the work by Camacho et al.~\cite{CamachoM19}, which proposes a SAT-based learning algorithm for LTL\textsubscript{f} formulas via Alternating Finite Automata as an intermediate representation.
Raha et al.~\cite{abs-2110-06726} present a scalable approach for learning formulas in a fragment of LTL\textsubscript{f} without the $\luntil$-operator, while Roy, Fisman, and Neider~\cite{0002FN20} consider the Property Specification Language~(PSL).
However, none of these works can exploit insights about the structure of the specification to aid the learning\slash{} mining process.


Finally, it is worth mentioning that LTL sketching can be seen as a particular case of syntax-guided synthesis, where syntactic constraints on the resulting formulas are expressed in terms of context-free grammars.
While this approach is more expressive in that it also allows restricting the syntax of the formulas that can be used to fill in missing parts of a sketch, we are convinced that the concept of sketching is more natural and intuitive, and it allows for efficient algorithms.
An example of a syntax-guided approach is SySLite~\cite{ArifLERCT20}, a CVC4-based tool for learning Past-time LTL (over finite executions).
However, we are unaware of any syntax-guided algorithm for our setting: learning LTL formulas over infinite executions.

\section{Preliminaries}
\label{sec:preliminaries}

We now formally introduce the basic notions that are used throughout the paper. 

\subsubsection{Finite and Infinite Words} 
To model trajectories of a system, we use the notion of words over an alphabet that represents the system events.
Formally, an \emph{alphabet} $\Sigma$ is a nonempty, finite set.
The elements of this set are called \emph{symbols}.
A \emph{finite word} over an alphabet $\Sigma$ is a sequence $u = a_0 \dots a_n$ of symbols $a_i \in \Sigma$ for $i \in \{0, \dots, n\}$.
The empty sequence, referred to as \emph{empty word}, is denoted by $\varepsilon$.
The length of a finite word $u$ is denoted by $|u|$, where $|\epsilon| = 0$.
Moreover, $\Sigma^*$ denotes the set of all finite words over the alphabet $\Sigma$, while $\Sigma^+$ = $\Sigma^* \setminus \varepsilon$ is the set of all non-empty words.

An \emph{infinite word} over $\Sigma$ is an infinite sequence $\alpha = a_0 a_1 \dots$ of symbols $a_i \in \Sigma$ for $i \in \nat$.
We denote the $i$-th symbol of an infinite word $\alpha$ by $\alpha[i]$ and the finite infix of $\alpha$
from position $i$ up to (and excluding) position $j$ with $\alpha[i,j)=a_i a_{i+1}\cdots a_{j-1}$.  
We use the convention that $\alpha[i,j)=\varepsilon$ for any $i\geq j$.
Further, we denote the infinite suffix starting at position $j\in \nat$ by $\alpha[j,\infty)=a_j a_{j+1} \cdots$.
Given $u \in \Sigma^+$, the infinite repetition of $u$ is the infinite word $u^{\omega} = uu \dots \in \Sigma^{\omega}$.
An infinite word $\alpha$ is called ultimately periodic if it is of the form $\alpha = uv^{\omega}$ for a $u \in \Sigma^\ast$ and $v \in \Sigma^+$.
Finally, $\Sigma^{\omega}$ denotes the set of all infinite words over the alphabet $\Sigma$.

\subsubsection{Propositional Boolean Logic}
Since the presented algorithms rely on the Satisfiability (SAT) problem, as a prerequisite, we introduce Propositional Logic.
Let $Var$ be a set of propositional variables, which take Boolean values from $\mathbb{B} = \{0,1\}$ (0 representing $\lfalse$ and 1 representing $\ltrue$).
Formulas in propositional (Boolean) logic---which we denote by capital
Greek letters---are inductively constructed as follows:
\begin{itemize}
	\item each $x \in Var$ is a propositional formula; and
	\item if $\Psi$ and $\Phi$ are propositional formulas, so are $\lnot \Psi$ and $\Psi \lor \Phi$.
\end{itemize}

Moreover, we add syntactic sugar and allow the formulas $\ltrue$, $\lfalse$, $\Psi \land \Phi$,$\Psi \Rightarrow \Phi$, and $\Psi \Leftrightarrow \Phi$, which are defined as usual.
A propositional valuation is a mapping $v : \mathit{Var} \to \mathbb{B}$, which maps propositional variables to Boolean values.
The semantics of propositional logic is given by a satisfaction relation $\models$ that is inductively defined as follows: $v \models x$ if and only if
$v(x) = 1, v \models \lnot\Phi$ if and only if $v \nvDash \Phi$, and $v \models \Psi \lor \Phi$ if and only if $v \models \Psi$ or $v \models \Phi$.
In the case that $v \models \Phi$, we say that $v$ satisfies $\Phi$ and call it a \emph{model} of $\Phi$.
A propositional formula $\Phi$ is \emph{satisfiable} if there exists a model $v$ of $\Phi$.
The \emph{size} of a formula is the number of its subformulas (as defined in the usual way).
The satisfiability problem of propositional logic is the problem of deciding whether a given formula is satisfiable.
Although this problem is well-known to be NP-complete, modern SAT solvers implement optimized decision procedures that can check satisfiability of formulas with millions of variables~\cite{BalyoHJ17}.
Moreover, virtually all SAT solvers also return a model if the input-formula is satisfiable.\\

\subsubsection{Linear Temporal Logic} 

Linear Temporal Logic (LTL) is a logic to reason about sequences of relevant statements about a system by using temporal modalities.
Formally, given set of propositions $\prop$ that represent relevant statements about the system under consideration, 
an LTL formula---which we denote by small Greek letters---is defined inductively as follows:
\begin{itemize}
	\item each proposition $p \in \prop$ is an LTL formula; and
	\item if $\psi$ and $\varphi$ are LTL formulas, so are $\lnot \psi$, $\psi \lor \varphi$, $\lnext \psi$ (``neXt''), and $\psi \luntil \varphi$ (``Until'').
\end{itemize}

As syntactic sugar, we allow the formulas
$\ltrue$, $\lfalse$, $\psi \land \varphi$ and $\psi \rightarrow \varphi$, which are defined as usual.
Additionally, we allow the temporal formulas
$\leventually \psi := \ltrue \luntil \psi$ (``Finally'')
and $\lglobally \psi := \lnot \leventually \lnot \psi$ (``Globally'').
Further, we use $\ltlset$ to denote the set of all LTL formulas using.

LTL formulas are interpreted over infinite words $\alpha\in(2^{\prop})^\omega$.
To define how an LTL formula is interpreted on a word, we use a valuation function $V$.
This function maps an LTL formula and a word to a Boolean value and is defined inductively as follows: 
\begin{align*}
&V (p, \alpha) = 1 \text{ if and only if } p \in \alpha(0),\\
&V (\lnot \varphi, \alpha) = 1 - V (\varphi, \alpha),\\
&V (\varphi \lor \psi, \alpha) = \max \{V (\varphi, \alpha), V (\psi, \alpha)\},\\
&V (\lnext \varphi, \alpha) = V (\varphi, \alpha[1,\infty))\\
&V (\varphi \luntil \psi, \alpha) = \max_{i \geq 0} \{ \min \{V (\psi, \alpha [i,\infty)), \min_{0 \leq j < i} \{ V (\varphi, \alpha[j,\infty))\}\}\}
\end{align*}
We call $V(\varphi, \alpha)$ the \emph{valuation of $\varphi$ on $\alpha$} and say that $\alpha$ \emph{satisfies} $\varphi$ if $V (\varphi, \alpha) = 1$.
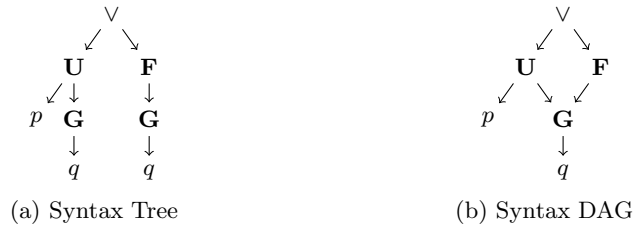
\begin{figure}
	\centering
	\subfloat[Syntax Tree\label{subfig:syntax-tree}]{
		\begin{minipage}{0.3\textwidth}
			\centering
			\begin{tikzpicture}
			\node (1) at (0, 0) {$\lor$};
			\node (2) at (-.5, -.7) {$\luntil$};
			\node (3) at (.5, -.7) {$\leventually$};
			\node (4) at (-1.0, -1.4) {$p$};
			\node (5) at (-0.5, -1.4) {$\lglobally$};
			\node (6) at (-0.5, -2.1) {$q$};
			\node (7) at (0.5, -1.4) {$\lglobally$};
			\node (8) at (0.5, -2.1) {$q$};
			\draw[->] (1) -- (2); 
			\draw[->] (1) -- (3);
			\draw[->](2) -- (4);
			\draw[->] (2) -- (5);
			\draw[->] (3) -- (7);
			\draw[->] (5) -- (6);
			\draw[->] (7) -- (8);
			\end{tikzpicture}
		\end{minipage}
	}\hspace{2cm}
	\subfloat[Syntax DAG\label{subfig:syntax-dag}]{
		\begin{minipage}{0.3\textwidth}
			\centering
			\begin{tikzpicture}
			\node (1) at (0, 0) {$\lor$};
			\node (2) at (-.5, -.7) {$\luntil$};
			\node (3) at (.5, -.7) {$\leventually$};
			\node (4) at (-1.0, -1.4) {$p$};
			\node (5) at (0, -1.4) {$\lglobally$};
			\node (6) at (0, -2.1) {$q$};
			\draw[->] (1) -- (2); 
			\draw[->] (1) -- (3);
			\draw[->](2) -- (4);
			\draw[->] (2) -- (5);
			\draw[->] (3) -- (5);
			\draw[->] (5) -- (6);
			\end{tikzpicture}
		\end{minipage}
	}
	\caption{Representation of LTL formulas}
	\label{fig:ltl-representation}
	\vspace{-0.4cm}
\end{figure}

For a precise representation of LTL formulas, we use syntax trees and syntax DAGs.
Syntax DAGs are directed acyclic graphs (DAG) that one obtains from a syntax tree by merging the common subformulas, providing a canonical representation of LTL formulas. 
Figure~\ref{fig:ltl-representation} illustrates the syntax tree and syntax DAG of the formula $(p\luntil \lglobally q)\lor \leventually\lglobally q$.

The size of an LTL formula $\size{\varphi}$ is defined as the number of its unique subformulas, which corresponds to the number of nodes in the syntax DAG of $\varphi$. For example, the size of the formula $(p\luntil \lglobally q)\lor \leventually\lglobally q$ is 6, as can be seen easily in Figure~\ref{subfig:syntax-dag}.

We denote the set of all LTL operators as $\operators=\prop\cup\unaryOp\cup\binaryOp$. 
Here, the propositions are the nullary operators, 
$\unaryOp=\{\neg,\lnext,\leventually,\lglobally\}$ are the unary operators and $\binaryOp=\{\lor,\land,\luntil\}$ are the binary operators of LTL.
Further, let $\ltlset$ denote the set of all LTL formulas.


\section{Problem Formulation}
\label{sec:problem}


Since the problem of \sketchprob{} relies heavily on LTL sketches, we begin with formalizing them first. 
\subsubsection{LTL Sketch}
An \emph{LTL sketch} is incomplete LTL formulas in which parts that are difficult to formalize can be left out.
The left-out parts are represented using placeholders, denoted by $?$'s as can be seen in Figure~\ref{fig:ltl-sketch}.
We comment on the superscripts on the placeholders in a moment.

\begin{figure}
	\centering
	\begin{tikzpicture}
	\node (1) at (0, 0) {$\placeholder^2$};
	\node (2) at (-.5, -.7) {$\luntil$};
	\node (3) at (.5, -.7) {$\placeholder^1$};
	\node (4) at (-1.0, -1.4) {$\placeholder^0$};
	\node (5) at (0, -1.4) {$\lglobally$};
	\node (6) at (0, -2.1) {$q$};
	\draw[->] (1) -- (2); 
	\draw[->] (1) -- (3);
	\draw[->](2) -- (4);
	\draw[->] (2) -- (5);
	\draw[->] (3) -- (5);
	\draw[->] (5) -- (6);
	\end{tikzpicture}
	\caption{An LTL sketch}
	\label{fig:ltl-sketch}
\end{figure}
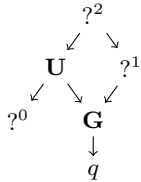

Formally, an LTL sketch $\sketch$ is simply an LTL formula whose syntax is augmented with placeholders.
The placeholders we allow can be of three types:
placeholders of arity zero referred to as Type-0 placeholders, that replace missing LTL formulas;
placeholders of arity one referred to as Type-1 placeholders, that replace missing unary operators; 
and placeholders of arity two referred to as Type-2 placeholders, that replace missing binary operators.
In the Figure~\ref{fig:ltl-sketch} (and in the rest of the paper), Type-$i$ placeholders are represented using $?^i$. 

Given (possibly empty) sets $\placeholderset^0$, $\placeholderset^1$ and $\placeholderset^2$ consisting of Type-0, Type-1 and Type-2 placeholders, respectively, we define LTL-sketches inductively as follows:
\begin{itemize}
	\item each element of $\prop \cup \placeholderset^0$ is an LTL-sketch; and
	\item if $\sketch_1$ and $\sketch_2$ are LTL-sketches, $\unaryvar\sketch_1$ is an LTL-sketch for $\unaryvar\in\unaryOp\cup\placeholderset^1$ and so is $\sketch_1\binaryvar\sketch_2$ for $\binaryvar\in\binaryOp\cup\placeholderset^2$.
\end{itemize}
Note that an LTL sketch in which $\placeholderset^0=\placeholderset^1=\placeholderset^2=\emptyset$ is simply an LTL formula.
Furthur, let $\placeholderset_{\sketch}=\placeholderset^0\cup\placeholderset^1\cup\placeholderset^2$
denote the set of all placeholders in an sketch $\sketch$.
For the sketch in Figure~\ref{fig:ltl-sketch}, $\placeholderset_\sketch=\{\placeholder^0, \placeholder^1, \placeholder^2\}$.
For brevity, in the rest of the paper, we refer an LTL sketch as a sketch. 

The placeholders are abstract symbols that apriori do not have meaning.
To assign meaning to a sketch, we need to substitute all Type-0 placeholders with LTL formulas, all Type-1 placeholders with unary operators and all Type-2 placeholders with binary operators.
We do this using a so-called substitution function (or substitution for short).

Formally, a \emph{substitution} function $s$ maps placeholders and operators present in a sketch to LTL operators and LTL formulas in such a way that: $s(\placeholder)\in\ltlset$ if $\placeholder\in\placeholderset^0$; $s(\placeholder)\in\Lambda_{U}$ if $\placeholder\in\placeholderset^1$; $s(\placeholder)\in\Lambda_{B}$ if $\placeholder\in\placeholderset^2$; and $s(\lambda)=\lambda$ for any LTL operator $\lambda\in\Lambda$.
Morever, a substitution $s$ is said to be \emph{complete} for an sketch $\sketch$ if 
$s$ is defined for every element in $\Lambda\cup\Pi_{\sketch}$ in $\sketch$.

A complete substitution $s$ can be applied to a sketch $\sketch$ to obtain an LTL formula.
To make this mathematically precise, we define a function $f_s$, which is defined recursively on the structure of $\sketch$ as follows:

\begin{align*}
f_s(\sketch_1  \mathbin{\placeholder^2} \sketch_2) &= f_s(\sketch_1) \mathbin{\circ} f_s(\sketch_2), \text{ where } \circ = s(\placeholder^2);\\
f_s(\mathbin{\placeholder^1} \sketch)) &= \unaryvar f_s(\sketch),\text{ where } \circ=s(\placeholder^1); \text{ and }\\
f_s(\placeholder^0) &= s(\placeholder^0).
\end{align*}

\subsubsection{Input sample} 
While there are many ways to fill a sketch, as alluded in the introduction, we rely on two finite, disjoint sets $P,N \subset (2^{\mathcal{\prop}})^{\omega}$ of ultimately periodic words, such that $P\cap N=\emptyset$.
The words in $P$, referred to as the positive examples, represent the desirable system executions that must be allowed by the resulting specification. 
On the other hand, the words in $N$, referred to as the negative examples, represent the undesirable system executions that must be disallowed by the resulting specification.
We call the pair $\sample = (P,N)$ a sample.
Since our sample consists of ultimately periodic words, we assume that they are stored as pairs $(u, v)$ of finite words, where $u \in (2^\prop)^{*}$ and $v \in (2^\prop)^{+}$.
Moreover, let the size of a sample be $|\sample|= \sum\limits_{uv^\omega\in P \cup N} |uv|$. 

We say that an LTL formula $\varphi$ is \emph{consistent} with a sample $\sample = (P,N)$ if $V (\varphi,uv^{\omega}) = 1$ for each $uv^\omega \in P$ (i.e., all positive words satisfy $\varphi$) and $V(\varphi,uv^{\omega}) = 0$ for each $uv^\omega \in N$ (i.e., all negative examples do not satisfy $\varphi$);

\subsubsection{The LTL sketching problem}
Having defined the setting, the main problem that we deal with in the paper is the following: 
\begin{problem}[\sketchprob]\label{prb:tl-sketching}
Given an LTL sketch $\sketch$ and a sample $\sample = (P,N)$, find a complete substitution $s$ for $\sketch$ such that $f_s(\sketch)$ is consistent with $\sample$ 
\end{problem}

Before we address the above (and the more general) problem, let us consider a simpler version of the problem:
\begin{problem}[\learnprob]\label{prb:tl-learning}
Given a sample $\sample = (P,N)$, find a minimal LTL formula $\varphi$ such that $\varphi$ is consistent with $\sample$ 
\end{problem}
The above problem is a restricted version of Problem~\ref{prb:tl-sketching}, where the sketch $\sketch=\placeholder^0$ (i.e., a single Type-0 placeholder).
Recently, it has attracted a lot of attention in the recent years due to its application in Explainable AI and Specification Mining.
Along with its theoretical analysis~\cite{abs-2102-00876}, there have been a number of efficient algorithms~\cite{abs-2110-06726,NeiderG18,CamachoM19} to solve it.

For the \learnprob{} problem, one can always find a solution due to the existence of a generic LTL formula that is consistent with a given sample, as indicated by the following remark:
\begin{remark}\label{rem:generic-formula}
Given sample $\sample$, there exists an LTL formula $\varphi$ such that $\varphi$ is consistent with $\sample$ and $|\varphi|$ is of size $\mathcal{O}(|S|^4)$.
\end{remark}
To construct such a formula, one needs to perform the following two steps. 
First, construct formulas $\varphi_{\alpha,\beta}$ for all $\alpha\in P$ and $\beta\in N$, such that $V(\varphi_{\alpha,\beta}, \alpha)=1$
and $V(\varphi_{\alpha,\beta}, \beta) = 0$, using a sequence of $\lnext$-operators and an appropriate propositional formula to describe the first symbol where $\alpha$ and $\beta$ differ.
Second, construct the generic consistent formula $\varphi= \bigvee_{\alpha\in P}\bigwedge_{\beta\in N} \varphi_{\alpha,\beta}$.

For the \sketchprob{} problem, however, one may not always find a solution.
This is illustrated using the following example.
\begin{example}\label{eg:formula-non-existence}
Consider the LTL-sketch $\lglobally(\placeholder^0)$ and a sample consisting of a single positive word $\alpha=\{p\}\{q\}^\omega\in P$ and a single negative word $\beta=\{q\}^\omega$.
For this LTL-sketch and sample, there does not exist any substitution that leads to an LTL specification consistent with the sample, which can be shown using contradiction. 
Towards contradiction, assume that there exists a specification $\lglobally(\varphi)$ such that $V(\alpha,\lglobally(\varphi))=1$ and $V(\beta,\lglobally(\varphi))=0$. 
Based on the semantics of $\lglobally$-operator, $V(\alpha,\lglobally(\varphi))=V(\alpha[1,\infty),\lglobally(\varphi))=1$.
On the other hand, $V(\beta,\lglobally(\varphi))=V(\alpha[1,\infty),\lglobally(\varphi))=0$ since $\beta=\alpha[1,\infty)$ leading to a contradiction.
\end{example}

Since there might not exist any complete substitution for a given LTL sketch and a sample, apriori, it is unclear whether the LTL sketching problem is solvable by any terminating algorithm.
However, we show in Section~\ref{sec:theoritical-results} that it is indeed possible to find a terminating algorithm for the \sketchprob{} problem. Later in Section~\ref{sec:algorithms}, we exploit this result to find suitable substitutions to sketches.


\section{Existence of a Complete Sketch}
\label{sec:theoritical-results}

To devise a terminating algorithm for the \sketchprob{} problem, we first introduce the related decision problem, which is the following:
\begin{problem}[\existprob]\label{prb:sketch-exist}
	Given an LTL-sketch $\sketch$ and a sample $\sample = (P,N)$, does there exist a complete substitution $s$ for $\sketch$ such that $f_s(\sketch)$ is consistent with $\sample$.
\end{problem}
In this section, we prove that this problem is indeed decidable and belongs to the complexity class \NP{}. 
Later we also provide a decision procedure based on SAT solving to decide the problem.

\subsection{The decidability result}\label{subsec:exists-prob-decidability}
For the decidability result, we introduce some concepts as a preparation. We begin by observing a key property of ultimately periodic words, which is as follows:
\begin{observation}\label{obs:repeating-suffix}
	Let $uv^\omega \in (2^\prop)^\omega$ and $\varphi$ be an LTL formula.
	Then, $uv^\omega[\abs{u} + i, \infty) = uv^\omega[\abs{u} + j, \infty)$ for $j \equiv i \mod \abs{v}$.
	Thus, $V(\varphi, uv^\omega[\abs{u} + i, \infty))=V(\varphi, uv^\omega[\abs{u}+j, \infty))$.
\end{observation}
Observation~\ref{obs:repeating-suffix} indicates that, for a word $uv^\omega$, there exists only a finite number of distinct suffixes of $uv^\omega$,
all of which originate in the initial $uv$ portion of $uv^\omega$.
We now define the set $\suf{uv^\omega}=\{uv^\omega[i,\infty)~|~ 0\leq i<\abs{uv} \}$ of all (possibly) distinct suffixes of a word $uv^\omega$.
Moreover, we define $\suf{\sample}=\bigcup_{uv^\omega\in(P\cup N)} \suf{uv^\omega}$ to be the set of all suffixes of words in $\sample$. 
Observation~\ref{obs:repeating-suffix} further indicates that, to determine the evaluation of an LTL formula $\varphi$ on an ultimately periodic word $uv^\omega$, it is sufficient to determine its evaluation on the initial $\abs{uv}$ suffixes of $uv^\omega$.

Thus, for a compact representation of the evaluation of $\varphi$ on $uv^\omega$, we introduce a table notation $\sattable{\varphi}{uv^\omega}$.
A table $\sattable{\varphi}{uv^\omega}$ is a $\abs{\varphi} \times \abs{uv}$ matrix that consists of the satisfaction of all the subformulas $\varphi'$ of $\varphi$ on the suffixes of $\{uv^\omega\}$.
We define the entries of the matrix as follows:
\begin{align*}
\sattable{\varphi}{uv^\omega}[\varphi',t] = V(\varphi', uv^\omega[t,\infty)) \text{ for all subformulas } \varphi' \text{ of } \varphi \text{ and } 0\leq t < \abs{uv}
\end{align*}


Based on the above definition of the table $\sattable{\varphi}{uv^\omega}$, we identify three properties of these tables. 
These properties form important building blocks of the decidability proof (i.e., proof of Theorem~\ref{thm:NP-exist}) as we see later.

The first property, which we refer to as \emph{Semantic} property, is that various rows of table are related to each other in a way that reflects the semantics of LTL.
We formalize a row of a table using the notation $\sattable{\varphi}{uv^\omega}[\varphi', \cdot]$, which refers to the row of $\sattable{\varphi}{uv^\omega}$ corresponding to subformula $\varphi'$.

Let us first demonstrate this property on a running example.
Consider the formula $\psi=p\vee \lnext q$ and the word $\alpha=\{p,q\}\{p\}\{q\}^\omega$.
Figure~\ref{fig:tables-for-p-and-Xq} illustrates the table $\sattable{\psi}{\alpha}$.
From the figure, one can see that the row $\sattable{\psi}{\alpha}[p\vee \lnext q,\cdot]$ corresponds to the bitwise-OR of the rows $\sattable{\psi}{\alpha}[p,\cdot]$ and $\sattable{\psi}{\alpha}[\lnext q, \cdot]$.
\begin{figure}
	\centering
	\setlength{\tabcolsep}{5pt}
	\renewcommand{\arraystretch}{1.25}
	\begin{tabular}{ | c | c | c | c |}
		\hline
		 & 0 & 1 & 2  \\
		\hline
		$p$ & 1 & 1 & 0\\
		\hline
		$q$ & 1 & 0 & 1\\
		\hline
		$\lnext q$ & 0 & 1 & 1 \\
		\hline
		$p\vee \lnext q$ & 1 & 1 & 1\\
		\hline      
	\end{tabular}
	\caption{Bit-vectors for the subformulas of $p\vee \lnext q$}
	\label{fig:tables-for-p-and-Xq}
\end{figure}

To formalize the relation between the rows corresponding to different subformulas $\varphi'$ of $\varphi$, we must uniquely identify the subformulas.
To this end, given an LTL formula $\varphi$, we assign unique identifiers $i\in\{1,\cdots, n\}$ to each node of the syntax DAG of $\varphi$.
This enables us to denote the subformula of $\varphi$ rooted at Node~$i$ with $\varphi[i]$.
For assigning identifiers, we follow the strategy that: we assign the root node with 1; and for every node, we assign an identifier smaller than its children, if it has any. 
Note that one can analogously assign labels to syntax DAGs of sketches.
Figure~\ref{fig:sketch-label} demonstrates identifiers for the formula $(p\luntil\lglobally q)\land(\leventually(\lglobally q))$.
We further define a function $\ell\colon\{1,\cdots,n\}\mapsto\Lambda$ that maps the identifiers to the corresponding operators in the syntax DAG.
\begin{figure}
	\centering
	\subfloat[Syntax DAG \label{subfig:sketch-dag}]{
		\begin{minipage}{0.3\textwidth}
			\centering
			\begin{tikzpicture}
			\node (1) at (0, 0) {$\land$};
			\node (2) at (-.5, -.7) {$\luntil$};
			\node (3) at (.5, -.7) {$\leventually$};
			\node (4) at (-1.0, -1.4) {$p$};
			\node (5) at (0, -1.4) {$\lglobally$};
			\node (6) at (0, -2.1) {$q$};
			\draw[->] (1) -- (2); 
			\draw[->] (1) -- (3);
			\draw[->](2) -- (4);
			\draw[->] (2) -- (5);
			\draw[->] (3) -- (5);
			\draw[->] (5) -- (6);
			\end{tikzpicture}
		\end{minipage}
	}\hspace{2cm}
	\subfloat[Labeling\label{subfig:sketch-label-label}]{
		\begin{minipage}{0.3\textwidth}
			\centering
			\begin{tikzpicture}
			\node (1) at (0, 0) {$1$};
			\node (2) at (-.5, -.7) {$2$};
			\node (3) at (.5, -.7) {$3$};
			\node (4) at (-1.0, -1.4) {$4$};
			\node (5) at (0, -1.4) {$5$};
			\node (6) at (0, -2.1) {$6$};
			\draw[->] (1) -- (2); 
			\draw[->] (1) -- (3);
			\draw[->](2) -- (4);
			\draw[->] (2) -- (5);
			\draw[->] (3) -- (5);
			\draw[->] (5) -- (6);
			\end{tikzpicture}
		\end{minipage}
	}
	\caption{Labeling for formula $(p\luntil\lglobally q)\land(\leventually(\lglobally q))$.}
	\label{fig:sketch-label}
\end{figure}

We now describe the set of equations that formalize the relation between the rows.
How a row $\sattable{\varphi}{uv^\omega}[\varphi[i],\cdot]$ corresponding to Node~$i$ relates to the others depends the operator $\ell(i)$ in the root node of $\varphi[i]$.
Thus, we list the relation separately for different LTL operators.

If $\ell(i)=p$ for some propostion $p$, then we have the following relation:
\begin{align}\label{eq:prop-table-semantics}
\sattable{\varphi}{uv^\omega}[\varphi[i],t] = 
\begin{cases}
1 \text{ if } p\in uv^\omega[t]\\
0 \text{ otherwise}
\end{cases}
\end{align}

If $\ell(i)$ is a unary operator and Node~$j$ is the left child of Node~$i$, we have the following relations:
\begin{align}
&\text{ if }\ell(i)=\neg:\ \sattable{\varphi}{uv^\omega}[\varphi[i],t] = 1- \sattable{\varphi}{uv^\omega}[\varphi[j],t] \ \text{ for } 0\leq t < \abs{uv}\label{eq:table-relation-neg}\\
&\text{ if }\ell(i)=\lnext:\ \sattable{\varphi}{uv^\omega}[\varphi[i],t] =
\begin{cases}
\sattable{\varphi}{uv^\omega}[\varphi[j],t+1] \ \text{ for } 0\leq t < \abs{uv}-1\\
\sattable{\varphi}{uv^\omega}[\varphi[j],\abs{u}] \ \text{ for } t=\abs{uv}-1
\end{cases}\label{eq:table-relation-next}
\end{align}
While Equation~\ref{eq:table-relation-neg} simply follows from the semantics of $\neg$-operator, Equation~\ref{eq:table-relation-next} for the $\lnext$-operator exploits Observation~\ref{obs:repeating-suffix} along with its semantics.
In particular, the entry $\sattable{\varphi}{uv^\omega}[\varphi[i],\abs{uv}-1]$ relies on the evaluation of $\varphi[j]$ on $uv^\omega[\abs{u},\infty)$.

If $\ell(i)$ is a binary operator, and Node~$j$ and Node~$j'$ are the left and right children of Node~$i$, respectively, then we have the following relation:
\begin{align}
&\text{ if }\ell(i)=\lor: \sattable{\varphi}{uv^\omega}[\varphi[i],t] = \sattable{\varphi}{uv^\omega}[\varphi[j],t] \lor \sattable{\varphi}{uv^\omega}[\varphi[j'],t] \text{ for } 0\leq t < \abs{uv}\label{eq:table-relation-lor}\\
&\text{ if }\ell(i)=\luntil:\ \sattable{\varphi}{uv^\omega}[\varphi[i],t] = \label{eq:table-relation-until}\\
&\begin{cases}
\bigvee_{t\leq t''<\abs{uv}}\Big[\sattable{\varphi}{uv^\omega}[\varphi[j'],t'']\wedge \bigwedge_{t\leq t'<t''}\sattable{\varphi}{uv^\omega}[\varphi[j],t']\Big] \text{ for } 0\leq t < \abs{u} \nonumber\\
\bigvee_{\abs{u}\leq t''<\abs{uv}}\Big[\sattable{\varphi}{uv^\omega}[\varphi[j'],t'']\wedge \bigwedge_{t'\in t\looparrowright_{u,v} t''}\sattable{\varphi}{uv^\omega}[\varphi[j],t']\Big] \text{ for } \abs{u}\leq t < \abs{uv} \nonumber
\end{cases}
\end{align}
Again, one can see that Equation~\ref{eq:table-relation-lor} follows from the semantics of the $\lor$-operator.
Equation~\ref{eq:table-relation-until} for the $\luntil$-operator
consists of two cases: the first case provides the relation for entries $t \in \{0,\cdots,\abs{u}-1\}$ in the
initial part u; the second case covers the entries
$t\in\{\abs{u},\cdots,\abs{uv}-1\}$ in the periodic part $v$. 
Thereby, the second case uses the periodic nature of $uv^\omega$ to ``loop back'' into the periodic part $v$ using the set $t\looparrowright_{u,v} t''$ defined as the follows:
\begin{align*}
t \looparrowright_{u,v} t'' = 
\begin{cases}
\{t,\cdots, t''-1\}  \text{ if } t < t''; \\
 \{\abs{u},\cdots, t''-1, t, \cdots, \abs{uv}-1 \} \text{ if } t'' \leq t
\end{cases}
\end{align*}





Having defined the Semantic property, let us now describe the second property, the \emph{Consistency} property. 
This property ensures that $\sattable{\varphi}{uv^\omega}[\varphi,0]=1$ if and only if $uv^\omega$ satisfies $\varphi$.
Thus, for an LTL formula $\varphi$ consistent with $\sample$, we have the following relation:
\begin{align}\label{eq:table-consistency}
	\sattable{\varphi}{uv^\omega}[\varphi,0]=1 \text{ for all } uv^\omega\in P, \text{ and } \sattable{\varphi}{uv^\omega}[\varphi,0]=0 \text{ for all } uv^\omega\in N
\end{align}

The final property we observe is called the \emph{Suffix} property. 
This property originates from the fact that, LTL, being a future-time logic, has the same evaluation on equal suffixes, i.e., $V(\varphi,u_1v_1^\omega[t,\infty)) = V(\varphi,u_2v_2^\omega[t',\infty))$ for $u_1v_1^\omega[t,\infty)=u_2v_2^\omega[t',\infty)$.
Formally, we state the property as follows:
\begin{align}\label{eq:table-suffix-property}	
\sattable{\varphi}{u_1{v_1}^\omega}[\varphi,t]=\sattable{\varphi}{u_2{v_2}^\omega}[\varphi,t'] \text{ for all } u_1v_1^\omega[t,\infty)=u_2v_2^\omega[t',\infty)
\end{align}
This property becomes significant later for constructing LTL formulas for Type-0 placeholders.


Having set up the prerequisites, we now proceed to provide an \NP{} algorithm to decide the \existprob{} problem.
For ease of presentation, we demonstrate the algorithm on a simple (but crucial) case in which $\sketch$ consists of only a single Type-0 placeholder $\placeholder^0$.
For this case, one might assume that non-deterministically guessing a substitution for the placeholder should suffice.
However, apriori one does not know the size of the LTL formula that is necessary to substitute the Type-0 placeholder.

Thus, in our $\NP$ algorithm, instead of guessing substitutions, we guess all the entries of the table $\sattable{\sketch}{uv^\omega}$ for each $uv^\omega$ in $\sample$.
Note that the tables have a finite dimension, precisely $\abs{\sketch}\times\abs{uv}$, for each word $uv^\omega$.
Thus, the overall process of simply guessing all the table entries can be done in time $\mathcal{O}({\poly(\abs{\sketch},\abs{\sample})})$.

After guessing the table entries, we must verify that the guessed tables satisfy the three properties, Semantic, Consistency, and Suffix, discussed earlier in this section.
It is easy to verify that checking the first two properties for the tables require time $\mathcal{O}(\poly(\abs{{\sketch}},{\abs{uv}}))$ for each $uv^\omega$ in $\sample$.
For checking the Suffix property, one must identify the equal suffixes in $\suf{\sample}$.
This fact that this can be also done in time $\poly(\abs{\sample})$, is a consequence of Lemma~\ref{lem:sketch-existence-formula}, stated below.
Intuitively, the result states that two suffixes are equal if they are equal only on a finite portion $b$ of size $\poly(\abs{u_1}, \abs{u_2}, \abs{v_1}, \abs{v_2})$. 
\begin{lemma}
	$u_1v_1^\omega[t,\infty)=u_2v_2^\omega[t',\infty)$ if and only if $u_1v_1^\omega[t,t+b)=u_2v_2^\omega[t',t'+b)$, where $b=\max(\abs{u_1[t,\abs{u_1})}, \abs{u_2[t',\abs{u_2})})+\lcm(\abs{v_1}, \abs{v_2})$.
\end{lemma}

\begin{proof}
	First, let us consider $u_1v_1^\omega[t,\infty)=u_2v_2^\omega[t',\infty)$. 
	Clearly, all prefixes of $u_1v_1^\omega[t,\infty)$ and $u_2v_2^\omega[t',\infty)$ are equal, i.e., $u_1v_1^\omega[t,t+b)=u_2v_2^\omega[t',t'+b)$ for all $b\in\nat$.
	
	For the other direction, we consider $u_1v_1^\omega[t,t+b)=u_2v_2^\omega[t',t'+b)$, for $b=\max(\abs{u_1[t,\abs{u_1})}, \abs{u_2[t',\abs{u_2})})+\lcm(\abs{v_1}, \abs{v_2})$.
	Without loss of generality, let us assume that $\abs{u_1[t,\abs{u_1})}\geq\abs{u_2[t',\abs{u_2})}$.
	To avoid clutter of notation, we denote $\mu=\abs{u_1[t,\abs{u_1})}$ and $\nu=\lcm(\abs{v_1}, \abs{v_2})$.
	Thus, in this case, $b=\mu+\nu$.
	
	The proof, now, is based on two main observations.
	First, we begin with the simple observation: 
	\begin{align*}
	u_1v_1^\omega[t,t+\mu) &= u_2v_2^\omega[t',t'+\mu); \text{ and }\\ u_1v_1^\omega[t+\mu,t+b) &= u_2v_2^\omega[t'+\mu,t'+b)
	\end{align*}
	Second, we have that 
	\begin{align*}
	(u_1v_1^\omega[t+\mu,t+b))^\omega &= v_1^\omega; \text{ and } \\
	(u_2v_2^\omega[t'+\mu,t'+b))^\omega &= (v_2^\omega[t'+\mu, t'+\mu+\abs{v_2}))^\omega
	\end{align*}
	The above observation is due to the fact that $u_1v_1^\omega[t+\mu,t+b) = v_1^{\kappa}$ for $\kappa=\nu/\abs{v_1}$ and $u_2v_2^\omega[t'+\mu,t'+b) = (v_2^\omega[t'+\mu, t'+\mu+\abs{v_2}))^\kappa$ for $\kappa=\nu/\abs{v_2}$.
	
	Now, combining the two observations, we have the following:
	\begin{align*}
	u_1v_1^\omega &= u_1v_1^\omega[t,t+\mu)\cdot (u_1v_1^\omega[t+\mu,t+b))^\omega\\ 
	&= u_2v_2^\omega[t',t'+\mu)\cdot (u_2v_2^\omega[t'+\mu,t'+b))^\omega\\ 
	&= u_2[t',\abs{u_2}) \cdot u_2v_2^\omega[\abs{u_2}, t'+\mu)\cdot (v_2^\omega[t'+\mu, t'+\mu+\abs{v_2}))^\omega\\ 
	&= u_2v_2^\omega
	\end{align*}
	\qed
\end{proof}

We now prove that if the guessed tables satisfy the three properties, then there exists a LTL formula $\psi$ that one can replace in the Type-0 placeholder to obtain a consistent LTL formula. 
This fact is asserted by the following theorem.
\begin{theorem}
	Let $\sample=(P,N)$ be a sample, $\sketch$ be a sketch with only Type-0 placeholders and tables  $\sattable{\sketch}{uv^\omega}$ be $\abs{\sketch}\times\abs{uv}$ matrices with $\{0,1\}$ entries for each $uv^\omega\in P\cup N$. Then, the following holds: the tables $\sattable{\sketch}{uv^\omega}$ satisfy the Semantic, Consistency and Suffix properties if and only if there exists a substitution $s$ such that LTL formula $f_s(\sketch)$ is consistent with $\sample$.
\end{theorem}

\begin{proof}
For simplicity, we again consider that $\sketch$ consists of only one Type-0 placeholder $\placeholder^0$. 
The proof can be seemlessly extended to multiple Type-0 placeholders.

	For the forward direction, we show the existence of the substitution $s$ by explicit construction of an LTL formula for  $\placeholder^0$.
Towards this, we first construct a sample $\sample'=(P',N')$ as follows: 
\begin{align*}
&P'= \{uv^\omega[t,\infty)\in\suf{\sample}~|~\sattable{\sketch}{uv^\omega}[\placeholder^0,t]=1,\ uv^\omega\in P\cup N,\ 0\leq i< \abs{uv}\}\\
&N'= \{uv^\omega[t,\infty)\in\suf{\sample}~|~\sattable{\sketch}{uv^\omega}[\placeholder^0,t]=0,\ uv^\omega\in P\cup N,\ 0\leq i< \abs{uv}\}.
\end{align*} 
Since the tables satisfy the Suffix property, we have that $P'\cap N' = \emptyset$.
We can now construct the generic LTL formula $\psi$ consistent with $\sample'$ based on Remark~\ref{rem:generic-formula}.
We claim that this formula $\psi$ can be substituted in $\placeholder^0$ to obtain a consistent LTL formula. 

Towards this, we first prove that $\sattable{f_s(\sketch)}{uv^\omega}[\psi,\cdot]=\sattable{\sketch}{uv^\omega}[\placeholder,\cdot]$ for all $uv^\omega\in P\cup N$.
To prove this, we exploit two simple observations.
First, using the definition of tables, we have $\sattable{f_s(\sketch)}{uv^\omega}[\psi,t] = V(\psi,uv^\omega[t,\infty))$ for each $uv^\omega[t,\infty)\in\suf{\sample}$.
Second, since $\psi$ is consistent with $\sample'$, we know $V(\psi,uv^\omega[t,\infty))=\sattable{\sketch}{uv^\omega}[\placeholder,t]$.
Together, we have $\sattable{f_s(\sketch)}{uv^\omega}[\psi,t] = \sattable{\sketch}{uv^\omega}[\placeholder,t]$.

Next, we prove that $\sattable{f_s(\sketch)}{uv^\omega}[f_s(\sketch)[i],\cdot]=\sattable{\sketch}{uv^\omega}[\sketch[i],\cdot]$ for each $0\leq i< \abs{\sketch}$ and word $uv^\omega\in P\cup N$. (Note that we denote the same nodes in $f_s(\sketch)$ and $\sketch$ using the same identifiers.)
Towards contradiction, we assume that there exists some $uv^\omega \in P\cup N$ and some $0\leq i< \abs{\sketch}$ and  such that $\sattable{f_s(\sketch)}{uv^\omega}[f_s(\sketch)[i],\cdot]\neq\sattable{\sketch}{uv^\omega}[\sketch[i],\cdot]$. 
Let $i^*$ be the maximum row for which the tables become unequal.
The proof, in general, will proceed via a case analysis on the operator $\ell(i^*)$ labeled in Node~$i^*$.
However, since for proof is similar for all the operators, we assume $\ell(i)=\neg$ and Node~$j^*$ is the left child of Node~$i^*$.
Recall that $j^*>i^*$ based on our assignment of identifiers.
Further, based on the Semantic property, $\sattable{f_s(\sketch)}{uv^\omega}[f_s(\sketch)[i^*],t] = 1 - \sattable{f_s(\sketch)}{uv^\omega}[f_s(\sketch)[j^*],t]$ and
$\sattable{\sketch}{uv^\omega}[\sketch[i^*],\cdot]= 1 - \sattable{\sketch}{uv^\omega}[\sketch[j^*],\cdot]$ for each $0\leq t\leq \abs{uv}$ (Equation~\ref{eq:table-relation-neg}).
This implies that $\sattable{f_s(\sketch)}{uv^\omega}[f_s(\sketch)[j^*],\cdot]\neq\sattable{\sketch}{uv^\omega}[\sketch[j^*],\cdot]$, contradicting the maximality of $i^*$.

Finally, observe that $\sattable{f_s(\sketch)}{uv^\omega}[f_s(\sketch),\cdot]=\sattable{\sketch}{uv^\omega}[\sketch,\cdot]$.
As a consequence, since tables $\sattable{\sketch}{uv^\omega}$ satisfy the Consistency property, so does tables $\sattable{f_s(\sketch)}{uv^\omega}$.
This implies that $f_s(\sketch)$ is consistent with $\sample$.

For the other direction, we construct tables $\sattable{\sketch}{uv^\omega}[\psi,\cdot]$ based on the tables $\sattable{f_s(\sketch)}{uv^\omega}[\psi,\cdot]$.
In particular, we have $\sattable{f_s(\sketch)}{uv^\omega}[\psi,\cdot] = \sattable{\sketch}{uv^\omega}[\psi,\cdot]$ for each $0\leq i<\abs{\sketch}$ and $uv^\omega \in P\cup N$.
Since tables $\sattable{f_s(\sketch)}{uv^\omega}[\psi,\cdot]$ satisfy the Semantic, the Consistency and the Suffix properties, so does the tables $\sattable{\sketch}{uv^\omega}[\psi,\cdot]$. \qed

\end{proof}


With this, we conclude the \NP{} algorithm for the case where $\sketch$ only has one Type-0 placeholder.
We can easily extend the algorithm to the case where $\sketch$ additionally consists of Type-1 and Type-2 placeholders.
In particular, we first guess the operators to be substituted for the Type-1 and Type-2 placeholders and substitute them.
We then obtain a sketch consisting of only Type-0 placeholders. 
We can now apply our algorithm that relies on guessing tables, as described above.
In total, we obtain the following result:
\begin{theorem}\label{thm:NP-exist}
	The \existprob{} problem is in $\NP$.
\end{theorem}

While we determine the complexity upper bound of the \existprob{} problem to be \NP{}, the lower bound is open and adds to the list of open problems in the area of LTL inference (see Section~6 of~\cite{abs-2102-00876}).

\subsection{The decision procedure}\label{subsec:exists-algo}
Based on the \NP{} algorithm described above, we now devise a decision procedure to decide the \existprob{} problem.
The decision procedure relies upon solving an instance of SAT problem to check if there exists suitable tables $\sattable{\varphi}{uv^\omega}$ which satisfy the three properties, discussed in Section~\ref{subsec:exists-prob-decidability}.

To encode entries of the tables, we first introduce the following variables: $y^{u,v}_{i,t}$ for each $i\in\{1, \cdots n\}$, $t\in\{0,\cdots,\abs{uv}-1\}$, and $uv^\omega\in P\cup N$.
A variable $y^{u,v}_{i,t}$ encodes the entry $\sattable{\varphi}{uv^\omega}[\varphi[i],t]$.
Further, to encode the operators to be substituted for the Type-1 and Type-2 placeholders in $\sketch$, we have the following variables:
\begin{enumerate*}[label=({\roman*})]
\item $x_{i,\lambda}$ for each Node~$i$ where $\ell(i)$ is a Type-1 placeholder and each $\lambda\in\Lambda_{U}$; and 
\item $x_{i,\lambda}$ for each Node~$i$ where $\ell(i)$ is a Type-2 placeholder and each $\lambda\in\Lambda_{B}$
\end{enumerate*}

We now impose constraints on the introduced variables to ensure that they encode tables that satisfy the properties for inferring a consistent LTL formula.
We achieve this by constructing a propositional formula $\existFormula$ using the introduced variables.
This formula ensures that the variables $y^{u,v}_{i,t}$ encode tables that satisfy the three required properties.

Since $\existFormula$ ensures the existence of a suitable tables, we have the following property of $\existFormula$: there exists a complete substitution $s$ for $\sketch$ such that $f_s(\sketch)$ is consistent with $\sample$ if and only if $\existFormula$ is satisfiable.
We can now simply check the satisfiability of $\existFormula$ using an off-the-shelf SAT solver to determine the existence of a complete substitution.

Internally,~$\existFormula=\Phi^{1,2}_{\placeholder}\wedge\Phi_{\mathit{sem}}\wedge \Phi_{\mathit{con}}\wedge\Phi_{\mathit{suf}}$ is a conjunction of four formulas.
The first conjunct $\Phi^{1,2}_{\placeholder}$ ensures that the Type-1 and Type-2 placeholders are substituted by appropriate operators.
The conjuncts $\Phi_{\mathit{sem}}$, $\Phi_{\mathit{con}}$ and $\Phi_{\mathit{suf}}$ ensure that the variables $y^{u,v}_{i,t}$ encode entries of tables that satisfy the Semantic  property (Equations~\ref{eq:table-relation-neg} to~\ref{eq:table-relation-until}), the Consistency property (Equation~\ref{eq:table-consistency}) and the Suffix property (Equation~\ref{eq:table-suffix-property}), respectively.
We now describe briefly how each of these formulas are constructed.

\paragraph{Constraints for $\Phi^{1,2}_{\placeholder}$.}
For each Node~$i$ labeled with a Type-1 placeholder (i.e., $\ell(i)\in\placeholderset^1$), we have the following constraint:
\begin{align}\label{eq:selective-label-uniqueness}
\Big[ \bigvee_{\lambda \in \Lambda_{U}} x_{i,\lambda} \Big] \land \Big[ \bigwedge_{\lambda \neq \lambda' \in \Lambda_{U}} \lnot x_{i, \lambda} \lor \lnot x_{i, \lambda'} \Big],
\end{align}
which ensures that the Type-1 placeholder is substituted by a unique unary operator.
For Type-2 placeholders, we have the exact same constraint except that the operators range from the set of binary operators $\Lambda_{B}$.
Now, we construct $\Phi^{1,2}_{\placeholder}$ as the conjunction of all such constraints for the nodes labeled with Type-1 and Type-2 placeholders.

\paragraph{Constraints for $\Phi_{\mathit{sem}}$.}
We define $\Phi_{\mathit{sem}} = \bigwedge_{uv^\omega\in P\cup N} \Phi^{u,v}$, where $\Phi^{u,v}$ is a propositional formula that ensures that the
variables $y^{u,v}_{i,t}$ satisfy the Equations~\ref{eq:table-relation-neg} to~\ref{eq:table-relation-until} for word $uv^\omega$ in $\sample$.
In $\Phi^{u,v}$, for instance, for each Node~$i$ labeled with $\lnext$-operator (i.e, $\ell(i)=\lnext$) and has Node~$j$ as its left child, we have the following constraint:
\begin{align}\label{eq:semantic}
\Big[\bigwedge_{0\leq t\leq \abs{uv}-1} \Big[ y^{u,v}_{i,t} \leftrightarrow
y^{u,v}_{j,t+1} \Big]\Big] \wedge \Big[y^{u,v}_{i,\abs{uv}-1} \leftrightarrow
y^{u,v}_{j,\abs{u}}\Big]
\end{align}
This constraints ensures that the variables $y^{u,v}_{i,t}$ satisfy Equation~\ref{eq:table-relation-next} for the word $uv^\omega$.
We construct similar constraints for the other operators based on their corresponding semantic relation. 

For nodes labeled with Type-1 and Type-2 placeholders, we additionally rely on variables $x_{i,\lambda}$ to determine the operator $\lambda$ to be substituted in Node~$i$.
For instance, for each Node~$i$ labeled with a Type-1 placeholder (i.e., $\ell(i)\in\placeholderset^1$) that has Node~$j$ as its left child, we have the following constraint:
\begin{align}\label{eq:placeholder-next-semantics}
x_{i,\lnext} \rightarrow \Big[ \bigwedge_{0\leq t\leq \abs{uv}-1} \Big[ y^{u,v}_{i,t} \leftrightarrow
y^{u,v}_{j,t+1} \Big]\Big] \wedge \Big[y^{u,v}_{i,\abs{uv}-1} \leftrightarrow
y^{u,v}_{j,\abs{u}}\Big]
\end{align}
The above constraint states that if Node~$i$ is substituted with a $\lnext$-operator (i.e., if $x_{i,\lnext}$ is true), the constraint on the $y^{u,v}_{i,t}$ variables is based on Equation~\ref{eq:table-relation-next}.
We construct similar constraints for when $x_{i,\lambda}$ is true for some other operator $\lambda$ based on the corresponding equation for $\lambda$.
Finally, we construct $\Phi^{u,v}$ as the conjunction of all such constraints that rely on the semantic equations.


\paragraph{Constraints for $\Phi_{\mathit{sem}}$.}
We construct $\Phi_{\mathit{con}}$ as follows:
\begin{align}
\Phi_{con} = \Big[\bigwedge_{uv^\omega\in P\cup N} \Phi^{u,v}\Big]\wedge \Big[\bigwedge_{uv^\omega\in P} y^{u,v}_{1,0} \wedge \bigwedge_{uv^\omega\in N} \neg y^{u,v}_{1,0}\Big]
\end{align}
This formula ensures that Equation~\ref{eq:table-consistency} is satisfied for the prospective bit-vectors, in addition to ensuring $\Phi^{u,v}$ holds for all $uv^\omega$ in $\sample$.

\paragraph{Constraints for $\Phi_{\mathit{suf}}$.}
In $\Phi_{\mathit{suf}}$, for each Node~$i$ labeled with a Type-0 placeholder (i.e., $\ell(i)\in\placeholderset^{0}$), we have the following constraint:
\begin{align}\label{eq:suffix}
\bigwedge_{u_1v_1^\omega[t,\infty)=u_2v_2^\omega[t',\infty)\in\suf{\sample}} \Big[ y^{u_1,v_1}_{i,t} \leftrightarrow y^{u_2,v_2}_{j,t'} \Big],
\end{align}
which ensures that the variables $y^{u,v}_{i,t}$ the entries of the tables satisfy Equation~\ref{eq:table-suffix-property}

The following lemma establishes the correctness of the decision procedure. 
\begin{lemma}\label{lem:sketch-existence-formula}
Let $\sketch$ be a sketch, $\sample$ a sample and $\existFormula$ the formula as defined above. Then, 
$\existFormula$ is satisfiable if and only if there exists a complete substitution $s$ such that $f_{s}(\sketch)$ is consistent with $\sample$.
\end{lemma}

\begin{proof}
	If $\existFormula$ is satisfiable, we can construct a suitable complete substitution $s$ based on the model $v$ of $\existFormula$.
	To begin with, for Node~$i$ labeled with a Type-1 or Type-2 placeholder $\placeholder$, we have $s(\placeholder) = \lambda$, where $\lambda$ is the unique operator for which $v(x_{i,\lambda})=1$.
	Next, for nodes labeled with Type-0 placeholders, we first construct tables $\sattable{\sketch}{uv^\omega}$ for each $uv^\omega\in P\cup N$ using the value of the variables $v(y^{u,v}_{i,t})$. 
	Based on the construction of $\existFormula$, all of these tables satisfy the Semantic, Consistency and Suffix properties. 
	Now, we use Lemma~\ref{lem:sketch-existence-formula} to find substitutions to Type-0 placeholders using the tables $\sattable{\sketch}{uv^\omega}$.
	
	For the other direction, we construct a satisfying assignment $v$ using the substitution function $s$ and tables $\sattable{f_s(\varphi)}{uv^\omega}$ for $uv^\omega\in P \cup N$.
	First, we assign $v(x_{i,\lambda}) = 1$ if and only if $s(\placeholder) = \lambda$ for a Node~$i$ labeled with a Type-1 or Type-2 placeholder $\placeholder$.
	Second, we assign $v(y^{u,v}_{i,t}) = \sattable{f_s(\varphi)}{uv^\omega}[f_s(\sketch)[i],t]$ for each $uv^\omega\in P\cup N$ and $0\leq t\leq |uv|$.
	This assignment $v$ satisfies $\Phi^{1,2}_{\placeholder}$, since we obtain $v$ from the syntax DAG of a valid LTL formula.
	Further, this assignment satisfies $\Phi_{\mathit{sem}}$, $\Phi_{\mathit{con}}$ and  $\Phi_{\mathit{suf}}$ because the tables satisfy Semantic, Consistency and Suffix properties, respectively, on which the constraints are based.
	Overall, $v$ is a model for $\existFormula$.\qed
\end{proof}

Finally, we have the following remark to assess the size of the encoding $\existFormula$.

\begin{remark}
Let $n = \abs{\sketch}$ and $m=\sum_{uv^\omega\in P\cup N} \abs{uv} $.
The formula $\existFormula$ ranges over $\mathcal{O}(n+nm)$ variables and is of size $\mathcal{O}(n + nm^3 + m^2)$.
\end{remark}


\section{Algorithms to complete an LTL sketch}
\label{sec:algorithms}
In this section, we describe two novel algorithms for solving the \sketchprob{} problem, which aim at searching for concise LTL formulas from sketches.
Our rationale behind finding concise specifications is motivated by two considerations: first, they are more human understandable and thus, easier for the engineers to interpret; second, potentially all verification and synthesis algorithms perform better with smaller specifications. 
To this end, our first algorithm relies on existing techniques for learning \emph{minimal} LTL formulas. 
Our second algorithm, alternatively, searches for formulas of increasing size using constraint solving.

\subsection{Algorithm based on LTL learning}\label{subsec:ltl-learn-algo}
Our first algorithm for finding concise LTL formulas builds upon the decision procedure for checking the existence of a complete substitution presented in Section~\ref{subsec:exists-algo}.
In particular, we exploit the formula $\existFormula$ to reduce the \sketchprob{} problem to a number of instances of the \learnprob{} problem, one for each Type-0 placeholder.
One can then exploit recent advancements~\cite{NeiderG18,Riener19} in solving the \learnprob{} problem to find concise ways of filling out the sketch. 

Algorithm~\ref{alg:ltl-learn} outlines the details of the algorithm. 
The first step is to construct $\existFormula$ from the given sample and sketch, as described in Section~\ref{subsec:exists-algo}.
If $\existFormula$ is unsatisfiable, the algorithm straight-away returns that no solution exists (Line~\ref{line:not-satisfiable}), as established by Theorem~\ref{thm:algo1-correctness}.
If it is satisfiable, we use a model (say $v$) of $\existFormula$, which can be obtained from any off-the-shelf SAT solver, to fill out $\sketch$.
We now describe in detail how to fill out a sketch.

Given a model $v$ of $\existFormula$, in a rather straightforward manner, one can substitute the Type-1 and Type-2 placeholders in $\sketch$ (Line~\ref{line:fill-type-1-2}).
For each Node~$i$ where $\ell(i)$ is a Type-1 and Type-2 placeholders, we assign $s(\ell(i))=\lambda$ where, $\lambda$ is the unique operator for which $v(x_{i,\lambda})=1$.

The Type-0 placeholders, however, are more challenging to substitute.
This is because, they represent entire LTL formulas.
Towards substituting Type-0 placeholders, for every Node~$i$ where $\ell(i)$ is a Type-0 placeholder (i.e., $\ell(i)\in\placeholderset^0$), we first construct a sample $\sample_i=(P_i,N_i)$ (Line~\ref{line:fill-type-0}) as follows: 
\begin{align*}
&P_i= \{uv^\omega[t,\infty)\in\suf{\sample}~|~v(y^{u,v}_{1,t})=1\}, \text{ and }\\
&N_i= \{uv^\omega[t,\infty)\in\suf{\sample}~|~v(y^{u,v}_{1,t})=0\}.
\end{align*}

Now, we learn a concise LTL formula $\varphi_i$ consistent with the sample $\sample_i$ (using e.g., algorithms by Neider and Gavran~\cite{NeiderG18}) for substituting $\ell(i)$ (Line~\ref{line:learn-ltl}).

\begin{algorithm}
	\caption{Algorithm based on LTL learning}\label{alg:ltl-learn}
	
	\begin{algorithmic}[1]
		\State \textbf{Input:} Sketch $\sketch$, Sample $\sample$
		\State Construct $\existFormula= \Phi^{1,2}_{\placeholder}\wedge \Phi_{\mathit{sem}}\wedge\Phi_{\mathit{con}}\wedge\Phi_{\mathit{suf}}$ 
		\If{$\existFormula$ is satisfiable (say with model $v$)}
		\State Substitute Type-1 and Type-2 placeholders in $\sketch$ using $v$\label{line:fill-type-1-2}
		\For{every $i$ such that $\ell(i)\in\placeholderset^0$}
		\State Construct $\sample_i=(P_i,N_i)$\label{line:fill-type-0}
		\State $\Phi_i\gets\mathit{Learn}(\sample_i)$\label{line:learn-ltl}
		\State Substitute Node~$i$ with $\Phi_i$ in $t$
		\EndFor
		\State \Return \sketch 
		\Else
		\State \Return LTL formula does not exist \label{line:not-satisfiable}
		\EndIf 		 
	\end{algorithmic}
\end{algorithm}

The correctness of the algorithm is established by the following theorem:
\begin{theorem}\label{thm:algo1-correctness}
Given sketch $\sketch$ and sample $\sample$,  Algorithm~\ref{alg:ltl-learn} terminates and completes $\sketch$ to output an LTL formula that is consistent with $\sample$, if such a formula exists.
\end{theorem}

Observe that Algorithm~\ref{alg:ltl-learn} constructs new samples for each Type-0 placeholder.
Each of these samples have size $\mathcal{O}(\size{\suf{\sample}})=\mathcal{O}(\size{\sample}^2)$,
which poses a challenge to the scalability of this algorithm.
Furthermore, the new samples are not optimized to produce the minimal possible substitutions.
Our next algorithm improves both on the runtime and the size of the inferred specification.

\subsection{Algorithm based on incremental SAT solving}\label{subsec:incremental-SAT-solving}
We, now, introduce an algorithm that is based on a series of SAT solving problems, similar to the SAT-based algorithm by Neider and Gavran~\cite{NeiderG18}.
Given a sample $\sample$ and a natural number $n\in\nat\setminus\{0\}$, we construct a propositional formula $\incrementalFormula{n}$ of size $\poly(n,\abs{\sample})$ that has the following two properties:
\begin{itemize}
	\item $\incrementalFormula{n}$ is satisfiable if and only if one can complete $\sketch$ to obtain an LTL formula of size $n$ that is consistent with $\sample$; and
	\item if $v$ is a model of $\incrementalFormula{n}$, then $v$ contains sufficient information to complete $\sketch$ to construct an LTL formula $\varphi$ of size $n$ that is consistent with $\sample$.
\end{itemize}

\begin{algorithm}
	\caption{Algorithm based on incremental SAT solving}\label{alg:syntax-tree}	
	\begin{algorithmic}[1]
		\State \textbf{Input:}  Sketch $\sketch$, Sample $\sample$
		\State Construct $\existFormula= \Phi^{1,2}_{\placeholder}\wedge\Phi_{\mathit{sem}}\wedge\Phi_{\mathit{con}}\wedge\Phi_{\mathit{suf}}$
		\If{$\existFormula$ is satisfiable}
		\State $n\gets\size{\sketch}-1$
		\Repeat
		\State $n\gets n+1$
		\State Construct $\incrementalFormula{n}= \Phi^{1,2}_{\placeholder}\wedge\Phi^{'}_{\mathit{sem}}\wedge\Phi_{\mathit{con}}\wedge\Phi^{0}_{\mathit{\placeholder},n}$
		\Until{$\incrementalFormula{n}$ is satisfiable (say with model $v$)}
		\State Substitute placeholders in $t$ using $v$
		\State \Return $\sketch$
		\Else
			\State \Return LTL formula does not exist
		
		\EndIf
		
	\end{algorithmic}
	\end{algorithm}

However, in contrast to the algorithms by Neider and Gavran, we first solve $\existFormula$ (discussed in Section~\ref{subsec:exists-prob-decidability}) to determine the existence of a complete substitution.
If and only if $\existFormula$ is satisfiable, our algorithm checks the satisfiability of $\incrementalFormula{n}$ for increasing values of $n$ (starting from 1) to search for an LTL formula of size at most $n$ that has the same syntactic structure as $\sketch$.
We construct the resulting LTL formula by substituting the placeholders in $\sketch$ based on a model $v$ of the satisfiable formula $\incrementalFormula{n}$.
This idea is illustrated in Algorithm~\ref{alg:syntax-tree}.


On a technical level, the formula $\incrementalFormula{n}$ is obtained by modifying certain parts of the formula $\existFormula$. 
Precisely, $\incrementalFormula{n} =  \Phi^{1,2}_{\placeholder}\wedge\Phi'_{\mathit{sem}}\wedge\Phi_{\mathit{con}}\wedge\Phi^{0}_{\mathit{\placeholder},n}$ employs a two modifications: a new formula $\Phi^{0}_{\mathit{\placeholder},n}$ replaces $\Phi_{\mathit{suf}}$; and $\Phi'_{\mathit{sem}}$ is obtained by adding more constraints to $\Phi_{\mathit{sem}}$.
The formula $\Phi^{0}_{\mathit{\placeholder},n}$ encodes the structure LTL formulas that substitute the Type-0 placeholders.
$\Phi'_{\mathit{sem}}$, again as in $\Phi_{\mathit{sem}}$, ensures that the variables $y^{u,v}_{i,t}$ encode table entries $\sattable{\varphi}{uv^\omega}[\varphi[i],]$ that satisfy Equations~\ref{eq:table-relation-neg} to~\ref{eq:table-consistency}.
We now briefly describe the constraints for the newly introduced formulas.

\paragraph{Constraints for $\Phi^{0}_{\mathit{\placeholder},n}$.}
The constraints here rely on an additional set of variables:
\begin{enumerate*}[label=({\roman*})]
\item $x_{i,\lambda}$ for each Node $i$ labeled with a Type-0 placeholder, and for each $i\in\{\abs{\sketch}+1,\cdots, n\}$ such that $\ell(i)$ is a Type-0 placeholder, and each $\lambda\in\Lambda$; and
\item $l_{i,j}$ and $r_{i,j}$ for each Node~$i$ labeled with a Type-0 placeholder, and each $i\in\{\abs{\sketch}+1,\cdots,n\}$ and $j\in\{\max(i,\abs{t}),\cdots,n\}$.
\end{enumerate*}
The variable $x_{i,\lambda}$, again, encode that Node~$i$ is labeled with $\lambda$.
The variables $l_{i,j}$ (and $r_{i,j}$) encode that the left child (and the right child) of Node~$i$ is Node~$j$.
Together the new variables encode the structure of the prospective LTL formulas for Type-0 placeholders.

We now impose constraints, similar to Constraint~\ref{eq:selective-label-uniqueness}, on the variables $x_{i,\lambda}$ to ensure each node is labeled by a unique LTL operator from $\Lambda$.
Additionally, we ensure that each Node~$i$ has a unique left child using the following constraint:
\begin{align}\label{eq:child-uniqueness}
\Big[ \bigvee_{\max(i,\abs{t})\leq j\leq n} l_{i,j} \Big] \land \Big[ \bigwedge_{\max(i,\abs{t})\leq j\neq j'\leq n } \lnot l_{i,j} \lor \lnot l_{i,j'} \Big],
\end{align}
We have a similar constraint to ensure the uniqueness of right child of a node.
Now, we construct $\Psi_{\mathit{str}}$ as the conjunction of all such structural constraints.

\paragraph{Constraints for $\Phi'_{\mathit{sem}}$.}
We rely on new variables $y^{u,v}_{i,t}$ for each Node~$i$ labeled with a Type-0 variables and each $i\in\{\size{\sketch}+1,n\}$, each $t\in\{0,\cdots,\abs{uv}-1\}$ and each $uv^\omega$ in $\sample$.
On these variables, we impose the following constraint, which is similar to Constraint~\ref{eq:placeholder-next-semantics}:
\begin{align}
\Big[x_{i,\lnext} \wedge l_{i,j}\Big] \rightarrow\bigwedge_{0\leq t\leq \abs{uv}-1} \Big[ y^{u,v}_{i,t} \leftrightarrow
y^{u,v}_{j,t+1} \Big] \wedge \Big[y^{u,v}_{i,\abs{uv}-1} \leftrightarrow
y^{u,v}_{j,\abs{u}}\Big],
\end{align}
that ensures that the $y^{u,v}_{i,t}$ variables encode entries of table that satisfy Equation~\ref{eq:table-relation-next}.
We construct $\Phi'_{\mathit{sem}}$ as the conjunction of $\Phi_{\mathit{sem}}$ and the new semantic constraints.

We, now, establish the correctness of Algorithm~\ref{alg:syntax-tree} using the following theorem:
\begin{theorem}\label{thm:algo1-correctness}
	Given sketch $\sketch$ and sample $\sample$,  Algorithm~\ref{alg:syntax-tree} terminates and completes $\sketch$ to output an LTL formula that is consistent with $\sample$, if such a formula exists.
\end{theorem}

While Algorithm~\ref{alg:syntax-tree} optimizes for the size of the inferred specification, it may not always return an minimal LTL formula. 
Example~\ref{eg:smaller-minimal-formula} demonstrates one such situation where Algorithm~\ref{alg:syntax-tree} could produce a sub-optimal result.
\begin{example}\label{eg:smaller-minimal-formula}
	Consider the sketch $\sketch=\placeholder_0 \lor \lnext\lnext p$ (Figure~\ref{fig:sketch-eg-cycle}) and sample $\sample$ consisting of positive examples $\{\}\{p\}\{\}^\omega$ and $\{\}\{\}\{p\}\{\}^\omega$ and a negative example $\{\}^\omega$.
	For this input, a possible output by Algorithm~\ref{alg:syntax-tree} is the formula 
	$\leventually p\lor \lnext\lnext p$ (Figure~\ref{fig:possible-formula}). 
	The minimal consistent formula $\lnext p\lor \lnext\lnext p$ (Figure~\ref{fig:minimal-formula}), however, is smaller.
	\begin{figure}[h]
		\centering
		\subfloat[Sketch\label{fig:sketch-eg-cycle}]{
			\begin{minipage}{3cm}
				\centering
				\begin{tikzpicture}
				\node (1) at (0, 0) {$\lor$};
				\node (2) at (-.5, -.7) {$\placeholder^0$};
				\node (3) at (0.5, -.7) {$\lnext$};
				\node (4) at (0.5, -1.4) {$\lnext$};
				\node (5) at (0.5, -2.1) {$p$};
				\draw[->] (1) -- (2); 
				\draw[->] (1) -- (3);
				\draw[->](3) -- (4);
				\draw[->] (4) -- (5);
				\end{tikzpicture}
			\end{minipage}
		}
		\subfloat[Possible formula\label{fig:possible-formula}]{
			\begin{minipage}{3cm}
				\centering
				\begin{tikzpicture}
				\node (1) at (0, 0) {$\lor$};
				\node (2) at (-.5, -.7) {$\leventually$};
				\node (3) at (0.5, -.7) {$\lnext$};
				\node (4) at (0.5, -1.4) {$\lnext$};
				\node (5) at (0.5, -2.1) {$p$};
				\node (6) at (-0.5,-1.4) {$p$};
				\draw[->] (1) -- (2); 
				\draw[->] (1) -- (3);
				\draw[->] (3) -- (4);
				\draw[->] (4) -- (5);
				\draw[->] (2) -- (6);
				\end{tikzpicture}
			\end{minipage}
		}
		\subfloat[Minimal formula\label{fig:minimal-formula}]{
			\begin{minipage}{3cm}
				\centering
				\begin{tikzpicture}
				\node (1) at (0, 0) {$\lor$};
				\node (2) at (0, -.7) {$\lnext$};
				\node (3) at (0, -1.4) {$\lnext$};
				\node (4) at (0, -2.1) {$p$};
				\draw[->] (1) -- (2); 
				\draw[->] (1) edge [bend right] (3);
				\draw[->] (2) -- (3);
				\draw[->] (3) -- (4);
				\end{tikzpicture}
			\end{minipage}	
		}
		\caption{Minimal formulas require sharing of nodes}
		\label{fig:not-minimal-formulas}
	\end{figure}
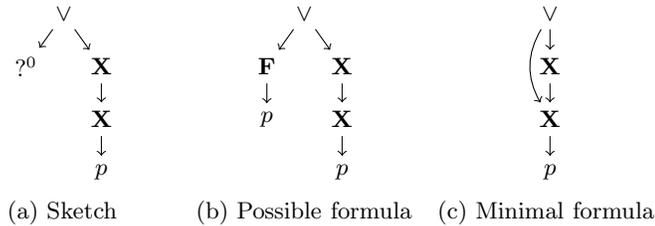
\end{example}
We leave the challenging problem of devising algorithms for searching for minimal LTL formulas from a given sketch as a part of future work.

\texttt{}\section{Experimental evaluation}
\label{sec:experiments}

In this section, we design experiments to answer the following research question: how do the presented algorithms perform in terms of their running times and the size of inferred specification?
Can the algorithms recover the specification intended by the engineer?

To answer these questions, we perform a comparative study of our algorithms (presented in Section~\ref{sec:algorithms}). 
Note that, to the best of our knowledge, none of the existing tools can solve the~\existprob{} problem the newly generated sample from which they learn is not ensured to produce the minimal formula.in its full generality.
Thus, we restrict ourselves to comparisons of the presented algorithms.

\pgfplotscreateplotcyclelist{mylist}{
	{red!90!black, mark=none},
	{green!50!black, mark=none},
	{blue!90!black, mark=none}
}
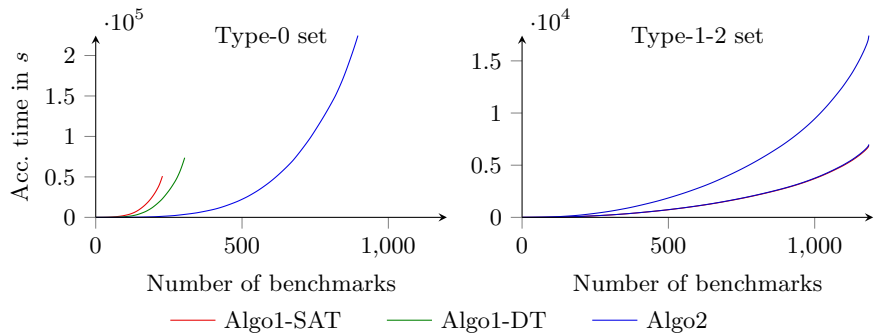
\begin{figure}
	\centering
	
	\begin{tikzpicture}
	\begin{groupplot}[
	width=6.25cm,
	height=4cm,
	group style={
		group size=2 by 1,
		ylabels at=edge left,
		horizontal sep=1cm
	},
	xlabel={Number of benchmarks},
	ylabel={Acc.\ time in $s$},
	axis lines=left,
	tick align=outside,
	title style={font=\footnotesize, at={(axis description cs:0.5, .8), anchor=north, inner sep=0pt}},
	cycle list name=mylist,
	axis on top,
	xmin=0,
	xmax=1200,
	]
	
	\nextgroupplot[title={Type-0 set}]
	\addplot table [x expr=\coordindex+1, y={Algo1SAT-type0}, col sep=comma] {csv/time-accumulated.csv}; \label{acc-time-plot:plot1}
	\addplot table [x expr=\coordindex+1, y={Algo1DT-type0}, col sep=comma] {csv/time-accumulated.csv}; \label{acc-time-plot:plot2}
	\addplot table [x expr=\coordindex+1, y={Algo2-type0}, col sep=comma] {csv/time-accumulated.csv};	 \label{acc-time-plot:plot3}
	
	
	\nextgroupplot[title={Type-1-2 set}]
	\addplot[red!80!black, mark=none] table [x expr=\coordindex+1, y={Algo1SAT-type12}, col sep=comma] {csv/time-accumulated.csv};
	\addplot[blue!80!black, mark=none] table [x expr=\coordindex+1, y={Algo1DT-type12}, col sep=comma] {csv/time-accumulated.csv};			
	\addplot[blue!80!black, mark=none] table [x expr=\coordindex+1, y={Algo2-type12}, col sep=comma] {csv/time-accumulated.csv};	
	
	\end{groupplot}
	\end{tikzpicture}
	\ref{acc-time-plot:plot1} Algo1-SAT \hskip 1.5em \ref{acc-time-plot:plot2} Algo1-DT \hskip 1.5em \ref{acc-time-plot:plot3} Algo2
	
	\caption{Accumulated runtime  of the algorithms on the two benchmark sets}
	\label{fig:acc-runtimes}
	\vspace{-0.5cm}
\end{figure}

We have implemented a prototype of our algorithms in a tool called \break \tool{} which is publicly available\footnote{\url{https://github.com/rajarshi008/LTL-sketcher}}.
While Algorithm~\ref{alg:ltl-learn} can exploit any LTL learning algorithm, we have chosen the algorithms presented by Neider and Gavran~\cite{NeiderG18} in their state-of-the-art LTL learning tool---\texttt{Flie}---one based on SAT solving and one on decision tree learning.
We refer to Algorithm 1 relying on the SAT-based algorithm as \algonesat{} and the one relying on decision tree learning as \algonedt{}.
We refer to Algorithm 2 as \algtwo{}.

For generating our benchmarks, we used the synthetic benchmark set presented by Neider and Gavran.
Their benchmark set consists of 1196 samples generated from 12 different LTL properties, which originate from a study by Dwyer et al.~\cite{DwyerAC98}.
Based on these samples, we generated two benchmark sets, which we refer to as the \zeroset{} set and the \onetwoset{} set.
Each benchmark in these two sets is simply a pair ($\sketch$, $\sample$) consisting of a sample and a sketch, which forms the input to our presented algorithms.

For generating a benchmark $(\sketch, \sample)$ in \zeroset{} set, we choose a sample $\sample$ generated from LTL formula $\varphi$ from the 1196 chosen samples; then construct a sketch $\sketch$ from $\varphi$ by substituting an aritrarily chosen subformula of size at least $\floor{\abs{\varphi}/2}$ by a Type-0 placeholder.
In this manner, we obtain 1196 pairs $(\sketch, \sample)$ for the \zeroset{} set.
We repeat the same process for the \onetwoset{} set, except that we construct a sketch $\sketch$ from $\varphi$ by substituting one aritrarily chosen operator by a Type-1 placeholder, if unary, otherwise with a Type-2 placeholder.
Here again we generate 1196 pairs $(\sketch, \sample)$ for the \onetwoset{} set.

All the experiments are conducted on a single core of a Debian machine with
Intel Xeon E7-8857 CPU (at 3~GHz) using up to 6~GB of RAM. The timeout
was set to be 1200~s for the run of each algorithm on each benchmark.

\begin{figure}
	\centering
	\begin{tikzpicture}	
	\begin{groupplot}[
	width=4cm,
	height=4cm,
	group style={
		group size=3 by 1,
		ylabels at=edge left,
		yticklabels at=edge left,
		horizontal sep=1.25cm
	},
	tick align=inside,
	title style={font=\footnotesize, at={(axis description cs:0.5, .8), anchor=north, inner sep=0pt}},
	cycle list name=mylist,
	axis on top,
	xmin=0,
	xmax=20,
	ymin=0,
	ymax=20,
	]
	
	\nextgroupplot[xlabel={Algo1-SAT}, ylabel={Algo1-DT}]
	\draw[gray] (rel axis cs:0, 0) -- (rel axis cs: 1, 1);
	\addplot[scatter=true,
	only marks,
	mark options={fill=gray},
	visualization depends on = {2*\thisrow{Freq-log1} \as \perpointmarksize},
	scatter/@pre marker code/.style={/tikz/mark size=\perpointmarksize},
	scatter/@post marker code/.style={}] table [col sep=comma,x={Algo1SAT1},y={Algo1DT1},meta index=2] {csv/sizes.csv};
	every mark/.append style={solid, fill=gray};	\label{size:plot1}

	\nextgroupplot[xlabel={Algo1-SAT}, ylabel={Algo2}, ylabel shift=-1mm]
	\draw[gray] (rel axis cs:0, 0) -- (rel axis cs: 1, 1);
	\addplot[scatter=true,
	only marks,
	mark options={fill=gray},
	visualization depends on = {2*\thisrow{Freq-log2} \as \perpointmarksize},
	scatter/@pre marker code/.style={/tikz/mark size=\perpointmarksize},
	scatter/@post marker code/.style={}] table [col sep=comma,x={Algo1SAT2},y={Algo22},meta index=2] {csv/sizes.csv};
	every mark/.append style={solid, fill=gray}; \label{size:plot2}
	
	\nextgroupplot[xlabel={Algo1-DT}, ylabel={Algo2}, ylabel shift=-1mm]
	\draw[gray] (rel axis cs:0, 0) -- (rel axis cs: 1, 1);			
	\addplot[scatter=true,
	only marks,
	mark options={fill=gray},
	visualization depends on = {2*\thisrow{Freq-log3} \as \perpointmarksize},
	scatter/@pre marker code/.style={/tikz/mark size=\perpointmarksize},
	scatter/@post marker code/.style={}] table [col sep=comma,x={Algo1DT3},y={Algo23},meta index=2] {csv/sizes.csv};
	every mark/.append style={solid, fill=gray}; \label{size:plot3}

	\end{groupplot}
	\end{tikzpicture}
	
	\caption{Comparison of sizes of resulting LTL formula on \zeroset{} set.}
	\label{fig:size-comparison}
	\vspace{-0.5cm}
\end{figure}
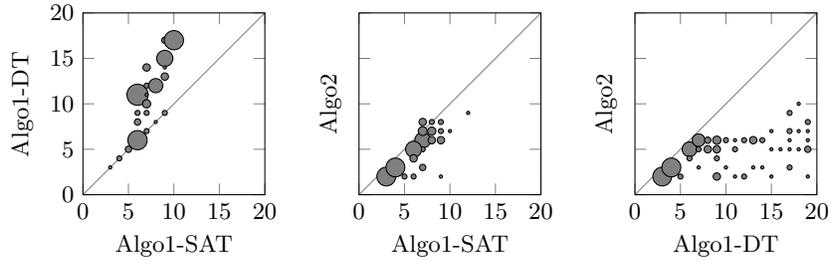


\subsubsection{Comparison of runtime}
First, we compare the three algorithms based on their runtimes on the \zeroset{} and the \onetwoset{} benchmark sets.
Figure~\ref{fig:acc-runtimes} presents the accumulated runtimes of the algorithms on the two benchmark sets.
From the figure, we observe that the runtimes of \algonesat{} and \algonedt{} is higher than that of \algtwo{} on the \zeroset{} set.
In fact, \algonesat{} and \algonedt{} experiences 957 and 881 time-outs, respectively, while, \algtwo{} only 289 time-outs. 
The superior performance of \algtwo{} can be attributed to the fact that \algonesat{} and \algonedt{} perform LTL learning on a newly generated sample $\sample'$ (as discussed in Section~\ref{subsec:ltl-learn-algo}) of size $\mathcal{O}({\abs{\sample}^2})$, which turns out to be an expensive procedure in most cases.

On the \onetwoset{} set, however, \algonesat{} and \algonedt{} display better runtime performance.
This is because for substituting Type-1 or Type-2 placeholders only, \algonesat{} and \algonedt{} do not generate new samples.
Rather, they only solve $\existFormula$ using a SAT-solver. 
\algtwo{}, on the other hand, additionally performs at least one iteration of its incremental SAT solving procedure.
This observation presents a possible future optimization of \algtwo{} to adapt based on the type of placeholder present in the sketch.

\subsubsection{Comparison of formula sizes}
We now compare the three algorithms based on the size of the resulting LTL formulas on the \zeroset{} set.
Figure~\ref{fig:size-comparison} presents a pair-wise comparison of the three algorithms in terms of size.
We notice that \algtwo{} produces the smallest formulas among all the algorithms.
This is because, \algtwo{} iteratively searches for consistent formulas of increasing size.
While \algonesat{} and \algonedt{} also try to optimize the size, the newly generated sample from which they learn is not ensured to produce the minimal formula.

\subsubsection{Comparison of formula recovery}
Finally, we compare the three algorithms based on whether they were able to recover the intended formula. In our experimental setup, we consider the intended formula to be the one that was used to generate the input sample.
Table~\ref{tab:intended-form} illustrates the percentage of runs in which the intended formula was recovered by the three algorithms on \zeroset{} and \onetwoset{} benchmark sets.
On \zeroset{} benchmark set, both \algonesat{} and \algonedt{} almost never produced the intended formula, while \algtwo{} produced the original formula in several runs. 
A possible explanation can be that \algtwo{} searches for simple completions first, finding the intended formula often. 
For Algo1, however, the intermediate samples generated will not be optimized to obtain the intended formulas. 
On \onetwoset{} benchmark set, on the other hand, all the algorithms recover the intended formulas in most runs.
This is because, the intermediate samples do not play a role in this benchmark set and thus, all the algorithms have a similar performance.

\begin{table}[h]
\begin{center}
	\renewcommand{\arraystretch}{1.2}
	
	\begin{tabular}{ | c | c | c | }
		\hline
		 & \zeroset{} & \onetwoset{} \\
		 \hline
		\algonesat{} & 0.1 & 97.4 \\  
		\algonedt{} & 0 & 97.4 \\
		\algtwo{} & 54.9 & 97.9 \\
		\hline
	\end{tabular}
\end{center}
\caption{Percentage of runs in which intended formula is recovered}
\label{tab:intended-form}
\end{table}

Overall, we conclude that Algorithm~\ref{alg:syntax-tree} performs better, in terms of runtime, size and recovery of specifications, even when Algorithm~\ref{alg:ltl-learn} uses state-of-the-art LTL learning techniques.
\vspace{-0.1cm}


\section{Conclusion and Future Work}
\label{sec:conclusion}
 
We have introduced LTL sketching, a novel way of writing formal specifications in LTL.
The key idea is that a user can write a partial specification, i.e., a sketch, which is then completed based on given examples of desired and undesired system behavior.
We have shown that the sketching problem is in \NP{} and presented two SAT-based sketching algorithms.
Our experimental evaluation has shown that our algorithms can effectively complete sketches consisting of different types of missing information.

A natural direction for future work is to lift the idea of specification sketching to other specification languages, such as Signal Temporal Logic~(STL)~\cite{DBLP:conf/formats/MalerN04}, the Property Specification Language~(PSL)~\cite{DBLP:series/icas/EisnerF06}, or Computation Tree Logic~(CTL)~\cite{ClarkeE81}.
We also plan to investigate how specification sketching can be applied to visual specifications, including UML (high-level) message sequence charts~\cite{DBLP:books/sp/03/HarelT03}.
Moreover, we intend to extend the notion of sketching beyond the use of examples to fill out placeholders (e.g., by allowing the engineer to constraint placeholders using simple logical formulas or regular expressions).

	\bibliographystyle{splncs04}
	\bibliography{bib}

\end{document}